\documentclass[final,onefignum,onetabnum]{siamart171218}

\usepackage{braket}
\usepackage{booktabs}
\usepackage{float}
\usepackage{graphicx} 
\usepackage{tabularx}
\usepackage{caption}
\usepackage{comment}
\usepackage{subcaption}
\usepackage{amsfonts}
\usepackage{graphicx}
\usepackage{algorithmic}
\usepackage{amsopn}

\ifpdf
  \DeclareGraphicsExtensions{.pdf,.png,.jpg}
\fi

\newsiamremark{remark}{Remark}
\newsiamremark{hypothesis}{Hypothesis}
\crefname{hypothesis}{Hypothesis}{Hypotheses}
\newsiamthm{claim}{Claim}

\headers{Quantum Preconditioning For Constrained Optimization Problems}{A. Ramesh, B. Sundar, M. Dupont and D.E. Bernal Neira}

\title{Quantum Preconditioning For Constrained Optimization Problems}

\author{Authors}

\makeatletter
\newcommand*{\addFileDependency}[1]{%
  \typeout{(#1)}%
  \@addtofilelist{#1}%
  \IfFileExists{#1}{}{\typeout{No file #1.}}%
}
\makeatother

\makeatletter
\def\refstepcounter@optarg[#1]#2{%
  \cref@old@refstepcounter{#2}%
  \cref@constructprefix{#2}{\cref@result}%
  \@ifundefined{cref@#1@alias}%
    {\def\@tempa{#1}}%
    {\def\@tempa{\csname cref@#1@alias\endcsname}}%
  \protected@edef\cref@currentlabel{%
    [\@tempa][\arabic{#2}][\cref@result]%
    \csname p@#2\endcsname\csname the#2\endcsname}%
}
\makeatother

\newif\ifshowcomments
\showcommentsfalse

\ifpdf
\hypersetup{
  pdftitle={Quantum Preconditioning For Constrained Optimization Problems},
  pdfauthor={Authors}
}
\fi

\author{
Anurag Ramesh\thanks{Rigetti Computing, 775 Heinz Avenue, Berkeley, CA 94710, USA, \\ Davidson School of Chemical Engineering, Purdue University, West Lafayette, IN 47907, USA
(\email{rames102@purdue.edu}).}
\and
Bhuvanesh Sundar\thanks{Rigetti Computing, 775 Heinz Avenue, Berkeley, CA 94710, USA
(\email{bsundar@rigetti.com}).}
\and
Maxime Dupont\thanks{Rigetti Computing, 775 Heinz Avenue, Berkeley, CA 94710, USA
(\email{mdupont@rigetti.com}).}
\and
David E. Bernal Neira\thanks{Corresponding Author. Davidson School of Chemical Engineering, Purdue University, West Lafayette, IN 47907, USA
(\email{dbernaln@purdue.edu}).}
}
\begin{document}

\newcommand{\anurag}[1]{\ifshowcomments\textcolor{blue}{A: #1}\fi}
\newcommand{\david}[1]{\ifshowcomments\textcolor{green}{D: #1}\fi}
\newcommand{\bhuvanesh}[1]{\ifshowcomments\textcolor{red}{B: #1}\fi}
\newcommand{\maxime}[1]{\ifshowcomments\textcolor{cyan}{M: #1}\fi}

\definecolor{responsecolor}{rgb}{0.55,0,0.55}
\newcommand{\response}[1]{\ifshowcomments\textcolor{responsecolor}{[\textbf{Response:} #1]}\fi}

\maketitle

\begin{abstract}
We study the effect of quantum preconditioning on constrained combinatorial optimization problems, focusing on balanced graph bi-partitioning.
The proposed approach uses two-point correlations between decision variables derived from the Quantum Approximate Optimization Algorithm (QAOA) to construct a modified objective function that is subsequently provided to mixed-integer programming (MIP) solvers.
The preconditioned MIP formulation retains the original hard constraint, and all incumbent solutions are evaluated under the original objective. Computational experiments on dense, weighted complete-graph instances show that the preconditioned problem instances reach near-optimal solutions faster, with most of the benefit already realized at the shallowest QAOA depth tested. Solver callback trajectories show this arises from earlier discovery of high-quality incumbents during the solution search.
These results support a hybrid optimization framework in which quantum algorithms provide problem-specific information to guide classical exact MIP solvers.
\end{abstract}



\section{Introduction}
\label{sec:intro}

Optimization under hard constraints is a fundamental challenge in applied mathematics and operations research.
In many settings, feasibility is not negotiable.
A valid solution must achieve a good objective value while satisfying structural constraints~\cite{Nocedal2006, NemhauserWolsey1988}.
Constrained combinatorial optimization~\cite{Korte2012, Papadimitriou1998} arises in a wide range of real-world applications, including network design~\cite{Saberi2023, Gavish1991}, vehicle routing~\cite{Cook2024LastMile, Dantzig1959, Helsgaun2017LKH3}, and parallel scientific computing, where graph partitioning underlies domain decomposition and sparse-matrix computations~\cite{Buluc2016GraphPartitioning,aboumrad2025accelerating}.
Mixed-integer programming (MIP)~\cite{Wolsey1998, Conforti2014} provides a general framework for formulating constrained combinatorial optimization problems, with hard constraints encoded explicitly in the formulation.
Modern exact MIP solvers~\cite{gurobi_website,Achterberg2009SCIP} primarily rely on the branch-and-bound algorithm~\cite{LandDoig1960}, augmented by complementary techniques such as primal heuristics~\cite{Berthold2006, Berthold_Lodi_Salvagnin_2025} and cutting-plane methods~\cite{Zhang2025Learning} to improve the efficiency of the solution search process to find good incumbents and also certify optimality.
Simulated annealing~\cite{Kirkpatrick1983} and parallel tempering~\cite{Earl2005ParallelTempering} provide alternative heuristic approaches, but do not by themselves provide the same exact certificate.

Mixed-integer programming problems are NP-hard in general~\cite{Papadimitriou1998, NemhauserWolsey1988}, although many structured instances are tractable in practice.
Consequently, despite substantial improvements in MIP solver performance driven by algorithmic advances~\cite{Koch2022Progress}, obtaining high-quality solutions or proving optimality remains challenging for large-scale instances involving hundreds or more decision variables.
Furthermore, their practical performance depends strongly on the problem structure and formulation~\cite{AchterbergKochTuchscherer2008}.
Even when the underlying optimal solution remains unchanged, a reformulation that guides the search toward high-quality incumbent solutions earlier can substantially alter solver behavior and performance.
A well-established strategy for MIP solvers is \emph{presolve}, a collection of pre-processing techniques applied before the branch-and-bound search begins with the goal of tightening the formulation, fixing variables, and reducing problem size~\cite{Andersen1995, Gamrath2015presolving}.
While \emph{presolve} is a standard component of state-of-the-art MIP solvers~\cite{gurobi_website, Achterberg2009SCIP}, it operates entirely on the structure of the given formulation.
A complementary direction is to use external information about the problem's solution structure to construct a modified objective that guides the solver toward high-quality incumbents earlier in the search. 

In this work, we investigate whether quantum algorithms can serve as a principled source of such external information for constrained combinatorial optimization problems, extending the quantum preconditioning framework of Dupont et al.~\cite{9prw-684p} to the setting of exact MIP solvers.
Figure~\ref{fig:workflow} summarizes this workflow.

\begin{figure*}[t]
    \centering
    \includegraphics[width=0.95\linewidth]{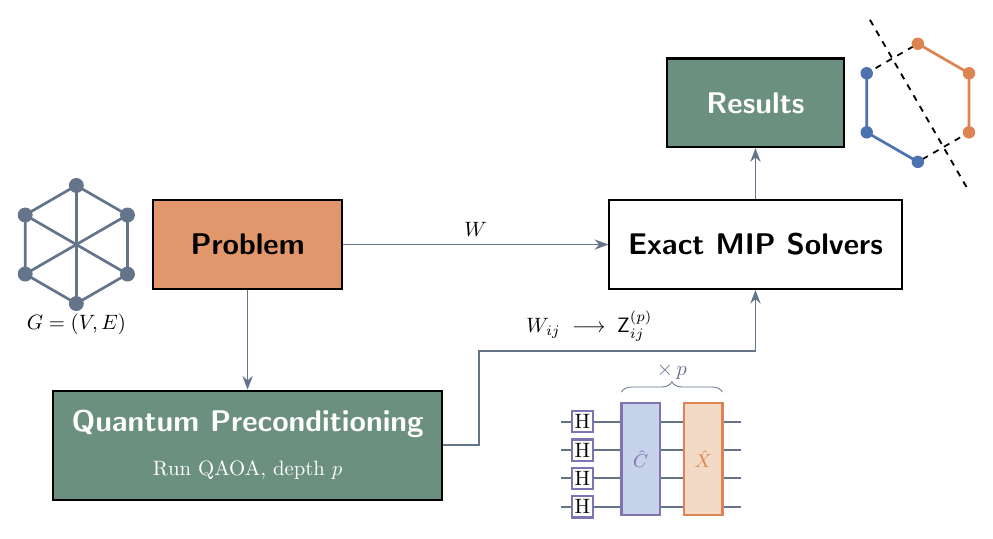}
    \caption{Overview of the quantum preconditioning workflow for constrained optimization problems using exact MIP solvers.
    A QAOA circuit of depth $p$ is optimized on the original weighted graph $G$ (with the balance constraint enforced as a soft penalty within the ansatz, Section~\ref{sec:quantum_preconditioning}) to produce a correlation matrix $\mathsf{Z}$.
    This matrix replaces the original weight matrix $W$ in the objective, and the resulting preconditioned problem is solved by exact MIP solvers such as \texttt{Gurobi} with the hard balance constraint retained.
    Every incumbent solution is re-evaluated under the original objective~\eqref{eq:cut_objective}.}
    \label{fig:workflow}
\end{figure*}

The remainder of the paper is organized as follows.
Section~\ref{sec:background} reviews the relevant quantum computing, quantum optimization, and quantum preconditioning background.
Section~\ref{sec:methods} introduces the optimization problem, develops the quantum preconditioning framework, and describes the experimental protocol.
Section~\ref{sec:results} presents the empirical results.
Section~\ref{sec:disc} discusses interpretation and limitations, and Section~\ref{sec:conclusions} concludes.

\section{Background}
\label{sec:background}


Quantum computing offers a fundamentally different paradigm for optimization by using quantum-mechanical phenomena to process information in ways unavailable to classical machines.
It has therefore attracted significant interest as a possible means of improving classical optimization approaches~\cite{Abbas2024, BERNALNEIRA2024108627}.
Quantum hardware has advanced rapidly from early noisy intermediate-scale quantum devices, whose usefulness is limited by noise and imperfect control~\cite{Preskill2018}, toward fault-tolerant operation~\cite{Preskill2025}.
A comprehensive review of quantum optimization is provided
in~\cite{Abbas2024, Koch2026QOBLIB}, covering algorithm design, benchmark selection, resource accounting, and comparison with classical methods as the central open questions for demonstrating practical benefit.

\subsection{Variational quantum algorithms}
Quantum approaches to combinatorial optimization are largely inspired by quantum adiabatic evolution~\cite{Farhi2000Adiabatic}, in which a quantum system initialized in the ground state of a simple Hamiltonian evolves toward one whose ground state encodes the target optimization problem.
In this context, the Hamiltonian is a quantum operator obtained by
encoding the classical cost function such that measuring the energy of
any computational basis state yields the corresponding objective value.
The optimal solution therefore corresponds to the ground state of this
operator.
Quantum annealing applies this principle on analog quantum hardware to search for low-energy configurations of optimization problems represented as Ising models~\cite{Finnila1994, Kadowaki1998, Das2008, Ising1925, Lucas2014}.
On programmable gate-based quantum hardware, variational quantum algorithms use a parameterized quantum circuit, consisting of a sequence of quantum gates, in a hybrid loop whose parameters are optimized classically to minimize the expectation value of a cost Hamiltonian~\cite{Cerezo2021}.
The Quantum Approximate Optimization Algorithm (QAOA) is one such variational quantum algorithm used to solve combinatorial optimization
problems. 
QAOA can be understood as a discrete, gate-based analog of quantum annealing. 
It replaces the continuous evolution of the Hamiltonian found in quantum annealing with $p$ discrete alternating layers, each applying a unitary generated by the problem Hamiltonian $\hat{C}$ (parametrized by angle $\gamma_\ell$) followed by a unitary generated by a mixing Hamiltonian $\hat{B} = \sum_{i \in V} \hat{X}_i$ (parametrized by angle $\beta_\ell$). 
The mixing Hamiltonian plays a role similar to the driver in quantum annealing.
It generates transitions between computational basis states, preventing the system from becoming trapped and allowing the algorithm to explore the solution space.
The angles $\{(\beta_\ell, \gamma_\ell)\}_{\ell=1}^p$ replace the fixed annealing schedule with a variational discrete path, optimized classically to minimize $\langle \hat{C} \rangle$~\cite{Farhi2014QAOA,
Blekos2023QAOAReview}.
In the limit $p \to \infty$ with appropriately chosen parameters, QAOA recovers the adiabatic limit and converges to the optimal solution~\cite{Farhi2014QAOA}.
There have also been several variants that build on these, such as recursive QAOA~\cite{bravyi2020obstacles, Bravyi2022, dupont2023quantum} and adaptive QAOA~\cite{zhu2022adaptive}.

Handling hard constraints is a central challenge in quantum optimization.
The most common approach is to encode constraints as soft penalty terms in the cost Hamiltonian~\cite{Lucas2014}, trading exact feasibility for a simpler circuit structure. 
An alternative is to design feasibility-preserving mixer operators that confine the quantum state evolution entirely within the feasible subspace, as in the quantum alternating operator ansatz~\cite{Hadfield2019QAOA} and its constraint-preserving variants~\cite{Fuchs2022}.
While this approach guarantees feasibility by construction, the resulting circuit depths and gate counts can be substantially higher than those of the standard QAOA mixer, limiting practical applicability on near-term quantum hardware.
A third strategy is warm-starting, in which a relaxed classical solution initializes the quantum circuit~\cite{Egger2021}, biasing the search toward feasible regions without explicitly enforcing constraints.
Each of these approaches involves a fundamental tradeoff between feasibility guarantees, circuit complexity, and solution quality.
The present work takes a complementary direction. 
Rather than modifying the quantum circuit to handle constraints, the balance constraint is enforced as a hard constraint in the exact MIP solver, and the quantum algorithm is used solely to extract correlation structure that reshapes the solution search space.

\subsection{Quantum preconditioning}

A distinct use of QAOA is to extract structural information about a
problem rather than to produce solutions directly.
The two-point correlations between pairs of decision variables, computed from an optimized QAOA circuit, encode information about promising variable alignments.
A strongly positive correlation between two variables indicates they tend to take the same value in low-energy states of the circuit, while a strongly negative correlation indicates they tend to take opposite values.
Recursive QAOA~\cite{Bravyi2022} exploits this idea by using single-variable expectation values to iteratively fix variables and reduce problem size. A related approach, quantum preconditioning, introduced by Dupont et~al.~\cite{9prw-684p}, uses two-point correlations to construct a correlation matrix that replaces the original weight matrix in the objective function.
The analytical motivation for this substitution is provided by~\cite{dupont2024extending}, which shows that QAOA two-point correlations can be used in a relax-and-round scheme with approximation guarantees matching the Goemans-Williamson algorithm~\cite{GoemansWilliamson1995} for the maximum cut problem on certain graph families, establishing that the correlation matrix encodes meaningful geometric information about the problem's solution structure.

The quality of preconditioning is controlled by the QAOA circuit depth $p$, with deeper circuits encoding richer correlation structure and the infinite-depth ($p\rightarrow \infty$) limit rendering the preconditioned problem to trivially recover the global optimal solution of the original problem~\cite{9prw-684p}.

In~\cite{9prw-684p}, the quantum preconditioning framework was evaluated on three problem classes: Sherrington-Kirkpatrick spin glasses~\cite{Sherrington1975}, random 3-regular graph maximum-cut problems~\cite{Dembo2017}, and a real-world power-grid energy optimization problem, using classical heuristics such as simulated  annealing~\cite{Kirkpatrick1983} and the Burer--Monteiro method~\cite{BurerMonteiro2003, BurerMonteiro2005} to solve both the original and preconditioned problems.
Dupont et al.~\cite{9prw-684p} found that quantum preconditioning let both heuristics converge faster on all three problem classes.
For random 3-regular maximum-cut instances, accounting for the time to generate the correlations still yielded a net quantum-inspired advantage through circuit emulation, the largest advantage window of the three cases.
Sherrington--Kirkpatrick spin glasses, preconditioned at $p=1$ with fixed tabulated angles, showed a narrower advantage window.
The real-world power-grid problem, preconditioned with variationally optimized angles, showed a window that widened from $p=1$ to $p=2$ and was also demonstrated experimentally on superconducting quantum hardware at $p=1$.

Here, we propose to extend the quantum preconditioning framework to hard-constrained problems.
By using quantum preconditioning with an exact MIP solver, we enforce the hard constraint with the MIP solver and investigate whether a similar performance improvement can be observed as well.

\subsection{Contributions}
The contributions of this paper are threefold.
First, we define a quantum preconditioning framework for the graph bi-partitioning problem based on QAOA-derived correlation matrices, including a rigorous theoretical analysis of the penalty parameter governing the balance constraint within the QAOA ansatz.
Second, we systematically study the effect of quantum preconditioning on the exact MIP solver behavior, focusing on improved incumbent solution discovery, runtime, and scaling.
Third, we provide detailed empirical evidence on dense all-to-all connected graph instances demonstrating regimes in which quantum preconditioning improves solver performance, along with an analysis of when and why such improvements may occur.

\section{Methods}
\label{sec:methods}

\subsection{Problem Definition}
\label{sec:problem}
In this work, the optimization problem under consideration is the balanced graph bi-partitioning problem, also known as minimum graph bisection.
It is important in applications such as network design~\cite{Saberi2023} and parallel scientific computing~\cite{Buluc2016GraphPartitioning}, where a graph must be divided into equally sized parts while minimizing communication across the partition.
The combination of a hard cardinality constraint and an NP-hard objective~\cite{GareyJohnsonStockmeyer1976} therefore provides a representative test case for studying the effect of quantum preconditioning.
We define the problem as follows: Let $G = (V, E)$ be a weighted undirected graph, where $V$ denotes the set of vertices and $E$ denotes the set of edges.
Each edge $(i,j) \in E$ is associated with a weight $w_{ij} \geq 0$.
The objective is to partition the vertex set $V$ into two disjoint subsets $V_1$ and $V_2$ such that the total weight of the edges crossing the partition is minimized, while ensuring that both subsets contain an equal number of vertices.

Following the standard Ising formulation~\cite{Lucas2014}, we introduce spin variables $z_i \in \{-1, +1\}$ for each vertex $i \in V$ where $z_i = +1$ indicates membership in $V_1$ and $z_i = -1$ indicates membership in $V_2$.
The problem takes the form
\begin{align}
\min_{\mathbf{z}} \quad & \frac{1}{2} \sum_{(i,j) \in E} w_{ij}
    \left(1 - z_i z_j \right) \label{eq:cut_objective} \\
\text{s.t.} \quad & \sum_{i \in V} z_i = 0, \quad
    z_i \in \{-1, +1\}, \quad \forall i \in V.
    \label{eq:balance_constraint}
\end{align}
Here, the term $\frac{1}{2}(1 - z_iz_j)$ equals $1$ when vertices $i$ and $j$ are in opposite partitions and $0$ otherwise, so the objective sums the total weight of edges crossing the partition.
The balance constraint~\eqref{eq:balance_constraint} enforces equal partition sizes and requires $|V|$ to be even.
To establish a classical baseline, the optimization problem~\eqref{eq:cut_objective}--\eqref{eq:balance_constraint} is solved directly using the commercial MIP solver \texttt{Gurobi}~\cite{gurobi_website}.
To this end, we transform the problem to an equivalent binary formulation using the substitution $x_i = (z_i + 1)/2 \in \{0, 1\}$, equivalently $z_i = 2x_i - 1$, so that $x_i = 1$ corresponds to $V_1$ and $x_i = 0$ to $V_2$.
Substituting into~\eqref{eq:cut_objective} and expanding gives
\begin{align}
\min_{\mathbf{x}} \quad & \sum_{(i,j) \in E} w_{ij}
    \left(x_i + x_j - 2 x_i x_j \right) \label{eq:cut_objective_binary} \\
\text{s.t.} \quad & \sum_{i \in V} x_i = \frac{|V|}{2}, \quad
    x_i \in \{0, 1\}, \quad \forall i \in V,
    \label{eq:balance_constraint_binary}
\end{align}
Formulation~\eqref{eq:cut_objective_binary}--\eqref{eq:balance_constraint_binary} is a binary quadratic program with a linear balance constraint.
\texttt{Gurobi} can solve this model directly.
Its \emph{presolve} may reformulate binary quadratic terms, and because each quadratic term contains a binary variable, it can construct an equivalent convex formulation without requiring the \texttt{NonConvex} parameter~\cite{gurobi_website}.
All solution quality evaluations throughout this work are performed against the original objective~\eqref{eq:cut_objective}.

\subsection{Quantum Preconditioning} 
\label{sec:quantum_preconditioning}
QAOA~\cite{Farhi2014QAOA} produces problem dependent correlation information, from which we construct a modified quadratic objective for the original problem that is then provided to the MIP solver.

The QAOA ansatz produces the quantum state
\begin{equation}
    \ket{\psi} =\left[\prod\nolimits_{\ell=1}^pe^{-i\beta_\ell\sum_j\hat{X}_j}e^{-i\gamma_\ell \hat C}\right] \ket{\psi_0},
    \label{eq:qaoa}
\end{equation}
where
\begin{equation}
    \hat{C} = \frac{1}{2}\sum_{(i,j) \in E} w_{ij}(1-\hat{Z}_i \hat{Z}_j) + \rho\!\left(\sum_{i \in V} \hat{Z}_i\right)^{\!2}
        \label{eq:Cost_Hamiltonian}
\end{equation}
is the quantized QAOA cost Hamiltonian with a soft penalty for constraint violations.
Here $\ket{\psi_0}=\ket{+}^{\otimes |V|}$ is the equal-superposition state, and $\hat X_j$ and $\hat Z_j$ are Pauli operators acting on qubit $j$~\cite{NielsenChuang2010}.
Each of the $p$ layers alternates a phase operation generated by $\hat C$ with a mixing operation generated by $\sum_j\hat X_j$.
The ansatz has $2p$ parameters $\{(\beta_\ell,\gamma_\ell)\in\mathbb{R}^2:1\leq\ell\leq p\}$, where $p$ is the QAOA depth and controls the quality of preconditioning~\cite{9prw-684p}.
The parameters are chosen to minimize $\braket{\hat C}$.
The nonnegative penalty $\rho\in\mathbb{R}^{+}$ controls how strongly the ansatz is biased toward balanced solutions.

For the optimized angles $\{(\beta_\ell,\gamma_\ell)\in\mathbb{R}^2:1\leq\ell\leq p\}$, quantum preconditioning constructs a modified objective from the two-point correlations of the state $\ket{\psi}$ as
\begin{equation}
\mathsf{Z}_{ij} = \delta_{ij} - \braket{\psi \vert \hat{Z}_i\hat{Z}_j \vert \psi}.
\label{eq:Correlation_Matrix}
\end{equation}
Here, the matrix $\mathsf{Z}$ defines the \textit{quantum-preconditioned} problem.
We emphasize that the penalty parameter $\rho$ enters only the QAOA cost Hamiltonian~\eqref{eq:Cost_Hamiltonian} as a soft representation of the balanced graph bi-partitioning constraint.
The balance requirement remains a hard constraint in the preconditioned problem actually solved using the exact MIP solver.

A useful extreme limit to understand the preconditioned problem is the $p \rightarrow \infty$ limit~\cite{9prw-684p}.
This case is neither a deployable algorithm nor a physically generated QAOA result, and is included solely for explanatory purposes.
At $p\rightarrow\infty$, the preconditioned problem is
\begin{equation}
Z_{ij}^{(\infty)} = \delta_{ij} - (z^\star)_i (z^\star)_j,
\label{eq:oracle_Z}
\end{equation}
where $\mathbf{z}^\star$ is the optimal solution to the original problem~\eqref{eq:cut_objective}--\eqref{eq:balance_constraint}.
For the problem studied here,
\begin{equation}
\tfrac12\mathbf z^\top\mathsf Z^{(\infty)}\mathbf z
=\tfrac12\left[n-(\mathbf z^\top\mathbf z^\star)^2\right].
\label{eq:oracle_identity}
\end{equation}
The squared overlap is at most $n^2$, with equality only when $\mathbf z=\pm\mathbf z^\star$.
The oracle objective is therefore globally minimized at those two configurations~\cite{9prw-684p}.

The total wall-clock time required to calculate \eqref{eq:Correlation_Matrix} constitutes the quantum-preconditioning time referenced in Section~\ref{sec:metrics} and is not included in the runtime measurements of MIP solver performance reported in Section~\ref{sec:results}.

\subsection{Solving the Preconditioned Problem}
\label{sec:solving_precond}

Once the correlation matrix $\mathsf{Z}$, defined in~\eqref{eq:Correlation_Matrix} is obtained, it is used to construct a preconditioned formulation of the original objective in~\eqref{eq:cut_objective}.

\begin{align}
\min_{\mathbf{z}} \quad & \tfrac{1}{2}\,\mathbf{z}^\top \mathsf{Z}\, \mathbf{z}
    = \sum_{(i,j) \in E} \mathsf{Z}_{ij}\, z_i z_j \label{eq:cut_objective_precond} \\
\text{s.t.} \quad & \sum_{i \in V} z_i = 0, \quad z_i \in \{-1, +1\},
    \quad \forall i \in V.\label{eq:balance_constraint_precond}
\end{align}
The diagonal entries vanish because $\mathsf Z_{ii}=1-\langle\hat Z_i^2\rangle=0$.
This is the quadratic Ising objective used in the analysis of Section~\ref{sec:quantum_preconditioning} where the matrix $\mathsf Z$ supplies the pairwise couplings.
The balance constraint~\eqref{eq:balance_constraint_precond} is unchanged from the original problem.
The preconditioned problem is then passed to \texttt{Gurobi}~\cite{gurobi_website} using the same binary transformation $x_i = (z_i + 1)/2$ described in Section~\ref{sec:problem} and is solved using the same solver settings as the baseline.

\subsection{Solver callback}
\label{sec:callback}
During the solve, we record every feasible solution reported through the \texttt{Gurobi} solver callback.
Concretely, at each incumbent event, the solver's current assignment is decoded from binary variables $x \in \{0,1\}^n$ to spins $z_i = 2x_i - 1$ and evaluated under the original objective~\eqref{eq:cut_objective}.
This callback-based evaluation is central to our analysis as it allows us to study not only the final outcome of a run, but also the time-dependent evolution of solution quality.
Therefore, this assessment provides mechanistic evidence that the preconditioned formulation genuinely improves the solver's search trajectory for the original problem.

\subsection{Numerical setup}
\label{sec:setup}

We evaluate the method on 50 randomly generated dense all-to-all connected graph instances with problem sizes $n \in \{12, 16, 20, \ldots, 40\}$ and edge weights sampled i.i.d.\ from $U(0, 2)$.
The same 50 instances are reused across all methods, QAOA depths, and penalty parameter values $\rho$.

The QAOA variational parameters $\{(\beta_\ell,\gamma_\ell)\}$ are optimized by minimizing $\braket{\hat C}$ with the Broyden--Fletcher--Goldfarb--Shanno (BFGS) quasi-Newton method~\cite{Nocedal2006}.
For $n\le 20$, we use $100$ random restarts per instance drawn uniformly from $\beta_\ell\in[0,\pi/2]$ and $\gamma_\ell\in[0,\pi/\sqrt n]$, retaining the best solution.
For $n>20$, the variational parameters are not re-optimized and instead transferred from the optimized values at $n_{\rm ref}=20$ using the rescaling $\beta_\ell(n)=\beta_\ell(20)$ and
$\gamma_\ell(n)=\gamma_\ell(20)\sqrt{20/n}$. 
This approach is motivated by the empirical observation that optimal QAOA parameters concentrate around problem-size-independent values when appropriately rescaled, and that parameters optimized for one instance can be successfully reused across structurally similar instances of different sizes with minimal performance degradation~\cite{Galda2021, shaydulin2023parameter, sundar2026qubit}.\footnote{The \(1/\sqrt{n}\) scaling of \(\gamma_\ell\) reflects the normalization of the cost Hamiltonian for dense problems with an extensive number of interactions. For a complete graph, the \(O(n^2)\) pairwise interaction terms are commonly rescaled by \(1/\sqrt{n}\) to maintain an appropriate extensive energy scale as the system size increases. Since QAOA evolves according to \(e^{-i\gamma_\ell H_C}\), this Hamiltonian normalization is equivalently implemented by scaling \(\gamma_\ell \propto 1/\sqrt{n}\). This scaling is also consistent with the large-\(n\) behavior of optimal QAOA angles observed for the Sherrington--Kirkpatrick model~\cite{Farhi2014QAOA, 9prw-684p}.}

Quantum preconditioning is performed for QAOA circuit depths $p \in \{1, 2, 3\}$.
For each depth $p$, we study a grid of penalty values spanning $\rho\in[0,2]$.
At $p=1$, we evaluate a tractable closed-form expression for the
expectation values in~\eqref{eq:Correlation_Matrix}~\cite{dupont2024extending}.
For $p>1$, we use numerically exact state-vector emulation.

All MIP solves are performed using \texttt{Gurobi 13.0} under two solver configurations: \emph{presolve}-enabled and \emph{presolve}-disabled.
The \emph{presolve}-disabled setting isolates the direct effect of quantum preconditioning on the branch-and-bound search, without the additional pre-processing routines applied by \texttt{Gurobi}.
Both configurations use 8 threads from a Linux cluster with 48 AMD EPYC 7643 2.3 GHz CPUs, an optimality gap tolerance of $10^{-6}$, and no wall-clock time limit.
All remaining \texttt{Gurobi} parameters (including \texttt{MIPFocus}, \texttt{Heuristics}, and \texttt{Cuts}) are left at their default values.
Solvers are terminated upon proving optimality or satisfying the gap tolerance.
Results under both configurations are reported throughout Section~\ref{sec:results}.

\subsection{Penalty parameter analysis}
\label{sec:penalty_analysis}

The penalty parameter $\rho$ in~\eqref{eq:Cost_Hamiltonian} controls how strongly the balance constraint $\sum_{i \in V} z_i = 0$ is enforced within the QAOA ansatz.
At $\rho = 0$, the ansatz concentrates purely on cut quality, and as $\rho$ increases, it is increasingly biased toward balanced configurations.
Before fixing a study range for $\rho$, we ask how large the penalty must be for the ansatz to actually favor balanced solutions in the adiabatic limit.

To make this precise, we analyze a penalized version of the original cut objective~\eqref{eq:cut_objective} where the hard balance constraint~\eqref{eq:balance_constraint} is incorporated as a quadratic soft penalty, as shown in~\eqref{eq:penalized_spin},
\begin{equation}
    C(\mathbf{z}) = H(\mathbf{z}) + \rho\!\left(\sum_{i \in V} z_i\right)^{\!2}.
    \label{eq:penalized_spin}
\end{equation}
Here, $H(\mathbf{z}) = \frac{1}{2}\sum_{(i,j)\in E}w_{ij}(1-z_iz_j)$ is the original cut objective with spin variables $z_i \in \{-1, 1\}$.
We call $\rho$ \emph{exactly sufficient} if every minimizer of~\eqref{eq:penalized_spin} satisfies $\sum_{i \in V} z_i = 0$, so that the penalty alone enforces feasibility as $p\rightarrow\infty$, without an explicit constraint.
Appendix~\ref{sec:penalty_proofs} proves the following deterministic exactness threshold for $\rho$.

\begin{theorem}[Deterministic exactness threshold]
\label{thm:det_threshold}
Let $G$ be any graph on $n = |V|$ vertices with weights $w_{ij} \in [0,2]$.
If $\rho > n/4$, then every minimizer of~\eqref{eq:penalized_spin} satisfies the balance constraint $\sum_{i \in V} z_i = 0$, and the minimizers of~\eqref{eq:penalized_spin} coincide exactly with those of the original constrained problem~\eqref{eq:cut_objective}--\eqref{eq:balance_constraint}.
This threshold is sharp: for every $\rho < n/4$ there exists a graph with weights in $U[0,2]$ whose every minimizer of~\eqref{eq:penalized_spin} is infeasible.
\end{theorem}

This bound is a worst-case statement, witnessed by weights in $U[0,2]$, and it grows linearly with $n$, reaching $10$ at $n=40$, far above the range $\rho \in [0,2]$ studied in this work.
Appendix~\ref{sec:penalty_proofs} shows that the gap is a genuine average-versus-tail effect rather than an artifact of the worst-case bound. Specifically, the appendix provides an almost-sure threshold under random weights that scales as $n/4$.
Per-instance computations show that the empirical critical penalty $\rho ^{\star }$ has a median near 0.35 and a 99th percentile below 0.90.
Finally, an expected-repair-step argument shows that the penalty favors the balance constraint in expectation once $\rho > 0.25$, independent of $n$.

Figure~\ref{fig:penalties_per_instance} shows the empirical distribution of $\rho^\star$ across problem sizes.
\begin{figure}[!t]
    \centering
    \includegraphics[width=0.5\linewidth]{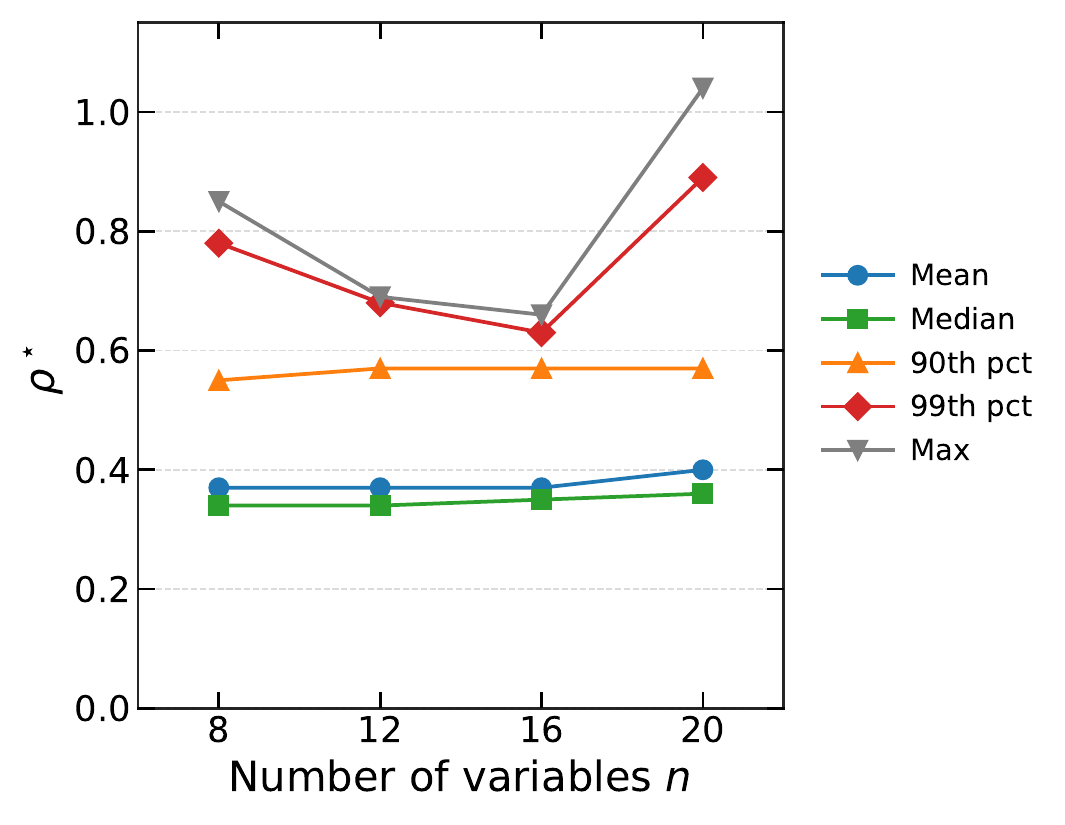}
    \caption{Empirical distribution of the per-instance critical penalty $\rho^\star$ for complete graphs with i.i.d. $U[0,2]$ weights, as a function of $n$, over the 50 instances}
    \label{fig:penalties_per_instance}
\end{figure}
Together, these results place the study grid $\rho \in [0,2]$ well below the worst-case threshold while covering nearly all of the empirical per-instance distribution, motivating the range used in the experiments below.

\subsection{Performance metrics}
\label{sec:metrics}

We use several complementary performance metrics to assess the effect of quantum preconditioning.

\paragraph{Approximation ratio}
For a feasible solution to the preconditioned problem with original-objective value $f$ and the corresponding baseline optimum $f^*$ for the same instance and \emph{presolve} setting, we define the approximation ratio as
\begin{align}
r = \frac{f}{f^*}
\end{align}
where $r = 1$ when the preconditioned solve recovers the baseline optimum exactly and $r > 1$ otherwise.

\paragraph{Quantum-preconditioning time}
The quantum-preconditioning time is the time required to produce the correlation coefficients of Equation~\eqref{eq:Correlation_Matrix}, and it depends on how the correlations are obtained.

In the present study, the correlations are computed classically. The correlations at $p=1$ are evaluated with the analytic expression of~\cite{dupont2024extending}, whose cost is polynomial in $n$, while the correlations at $p=2,3$ are obtained by exact state-vector emulation of the depth-$p$ QAOA circuit.
Because the instances are all-to-all connected, the emulation evolves the full $n$-qubit wavefunction, whose $2^n$ amplitudes make its time and memory scale exponentially in $n$.
This exponential cost of exact emulation, rather than any property of the preconditioner, restricts state-vector results to small $n$ and shallow depth.

On quantum hardware, correlations are estimated from measurement shots, and the cost is determined by the circuit. 
A depth-$p$ QAOA circuit on the dense $n$-variable problem, compiled onto a linear qubit chain through a swap network, has depth $O(pn)$~\cite{dupont2025benchmarking}.
Therefore, each shot runs in time $O(pn)$ and estimating the correlations from $N_{\rm ex}$ shots scales as
\begin{align}
t_{\rm precond} \sim O\!\left(p\,n\,N_{\rm ex}\right),
\label{eq:preconditioning_time_model}
\end{align}
with system-dependent prefactors for gate durations, measurement and reset latency, and classical overhead.
Since the time for the exact MIP solver to reach the $\epsilon$-optimal threshold grows exponentially in $n$, the solver runtime is the dominant cost at large problem sizes.

\paragraph{Runtime}
The runtime is the total wall-clock time until solver termination, including the time required to certify optimality or otherwise conclude the solve.

\paragraph{Time to an $\epsilon$-optimal threshold}
Let $f_s^*$ denote the baseline optimum for a problem instance $s$.
For a tolerance $\epsilon \geq 0$, the time to the $\epsilon$-optimal threshold is
\begin{equation}
    t_{\epsilon,s}:=\inf\left\{t\geq0:
    f_{\mathrm{orig},s}(t) \leq f_s^* + |f_s^*|\,\epsilon\right\}.
    \label{eq:eps_hit}
\end{equation}
The special case $\epsilon = 0$ is the first time the preconditioned solve reaches the baseline optimum exactly.
For $\epsilon > 0$, the metric records when the solve first reaches a solution within the prescribed relative tolerance.
We report two values of the tolerance: $\epsilon = 0$ (exact recovery of the baseline optimum) and $\epsilon = 0.01$ (a solution within $1\%$ of the baseline optimum).

\paragraph{Penalized average time}
For each run $s\in \mathcal S$, we define the penalized time
\begin{equation}
T_s=
\begin{cases}
t_{\epsilon,s}, & \text{if run $s$ reaches the threshold},\\
T_{\mathrm{run},s}+\alpha T_{\mathrm{base},s},
& \text{otherwise}.
\end{cases}
\qquad
\overline T=\frac1{|\mathcal S|}\sum_{s\in\mathcal S}T_s .
\label{eq:PAR}
\end{equation}
The penalized average time is $\overline T$.
Here $T_{\mathrm{run},s}$ is the runtime of the preconditioned solve, $T_{\mathrm{base},s}$ is the baseline runtime for the same instance, and $\alpha$ is a scalar penalty factor.
We use $\alpha=1$, so a successful run contributes its observed time to the threshold and an unsuccessful run contributes its runtime plus the cost of one baseline solve.
We report the success fraction alongside the penalized average time.

\section{Results}
\label{sec:results}

\subsection{Callback trajectory analysis}
\label{sec:callback_traj}
To understand how quantum preconditioning affects the branch-and-bound search, we analyze the feasible solutions reported through the \texttt{Gurobi} solver callback.
For every callback solution, the assignment obtained from the preconditioned problem is re-evaluated under the original objective.
Figure~\ref{fig:trajectory} shows the resulting incumbent trajectory for a single instance at $n = 40$ and QAOA depth $p = 1$ across different values of $\rho$ with \emph{presolve} enabled.
We find that for all tested values of $\rho$, the preconditioned problem quickly reaches a near-optimal incumbent.
However, finding the baseline optimum as an incumbent takes longer and depends sensitively on $\rho$.

\begin{figure*}[t]
    \centering
    \includegraphics[width=0.9\linewidth]{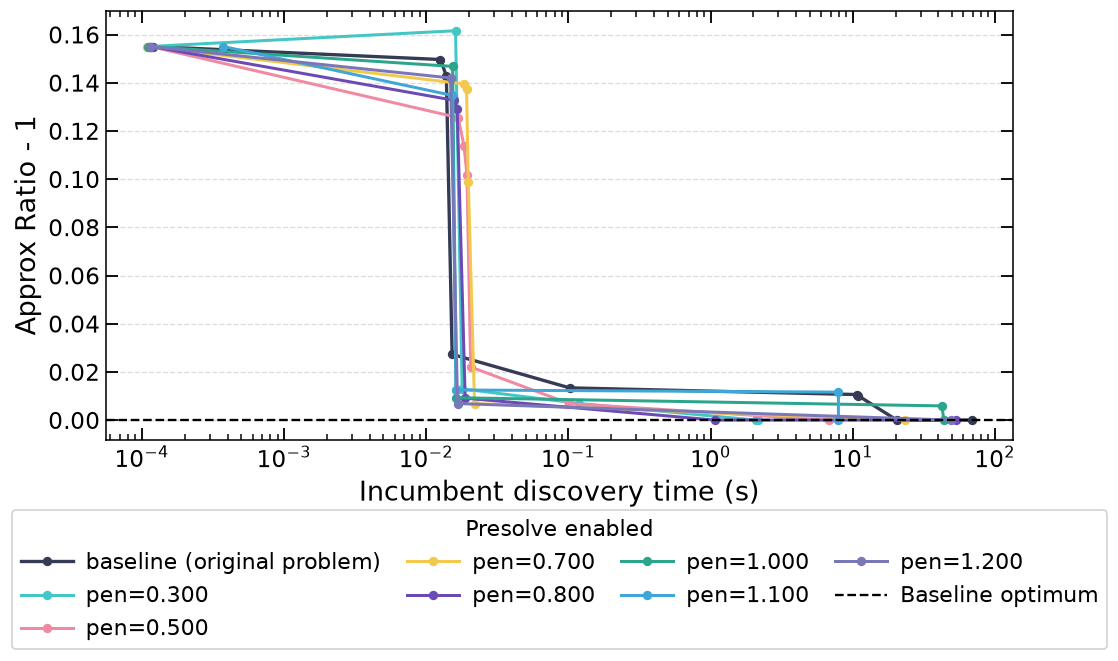}
    \caption{Incumbent solution trajectory for a single illustrative instance at $n = 40$ and QAOA depth $p = 1$ at select values of the penalty parameter $\rho$.
    Each callback solution found while solving the preconditioned problem is re-evaluated using the original objective~\eqref{eq:cut_objective}.}
    \label{fig:trajectory}
\end{figure*}

\subsection{Time to optimal and near-optimal thresholds}
\label{sec:epstime}

Figure~\ref{fig:PAR_eps_0} reports the time required to reach exact-recovery ($\epsilon = 0$) and near optimal ($\epsilon = 0.01$)  thresholds for $n = 40$ and QAOA depth $p = 1$, as a function of the penalty parameter $\rho$, across 50 independently generated instances, with \emph{presolve} enabled and disabled.
The dashed lines show the penalized average time~\eqref{eq:PAR}, denoted PAR in the figure, which applies a time penalty only to the instances that fail to reach the threshold, while successful instances contribute their observed time unchanged.
We fix $n = 40$ and $p = 1$ because this is the largest problem size and the shallowest circuit depth studied and provides a conservative estimate.

We find a speed-reliability tradeoff under both \emph{presolve} configurations.
Smaller penalty values give faster successful runs but more failures, and the failure count drops sharply as $\rho$ increases before leveling off at a small residual number of instances that never fully disappears across the tested range. 
This residual count is not a limit of solving time, since every solve in this dataset terminates with \texttt{Gurobi} reporting a certified optimal solution.
A failure to reach the $\epsilon = 0$ threshold therefore means that the preconditioned problem never finds the optimal solution to the original problem as an incumbent solution.

We observe that among successful instances, the mean time to exact recovery ($\epsilon = 0$) rises with $\rho$.
Since the number of successful instances grows with $\rho$, the rise could reflect more instances joining the pool.
In addition to exact recovery, we also observe that for a threshold of $\epsilon=0.01$, all 50 instances succeed at every value of $\rho$, and yet the mean time still grows with $\rho$.

Disabling \emph{presolve} affects the two problems very differently. The baseline time to the $\epsilon = 0$ threshold (indicated by the horizontal line) increases sharply once \emph{presolve} is disabled, while the time to that same threshold for the preconditioned problem stays close to its \emph{presolve}-enabled values across every tested $\rho$.
We also observe that disabling \emph{presolve} affects the PAR curve for the preconditioned problem, specific to exact recovery.
This follows from the fact that PAR adds the baseline time to exact recovery as the time penalty for every unsuccessful instance (see Section~\ref{sec:metrics}), shifting the curve further up.

At the near-optimal threshold $\epsilon = 0.01$, every instance succeeds at every tested $\rho$ under both \emph{presolve} settings, and the failure-driven sensitivity to $\rho$ disappears with it.
The aggregate scaling analysis of Section~\ref{sec:scaling} evaluates whether this benefit persists across all $n$ and $p$.

\begin{figure*}[t]
    \centering
    \includegraphics[width=\linewidth]{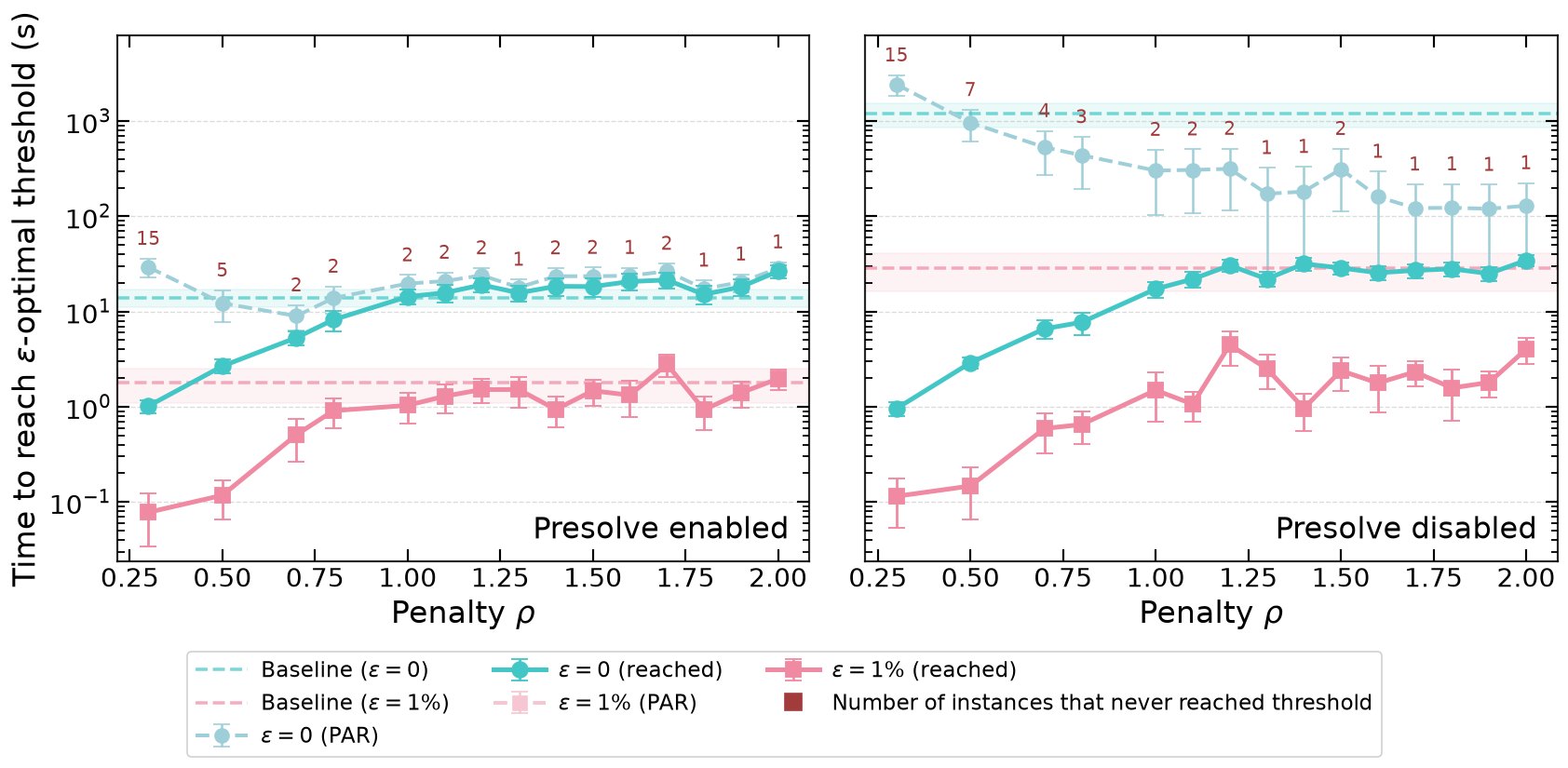}
    \caption{Time to the $\epsilon$-optimal threshold for $n=40$ and QAOA depth $p=1$, as a function of penalty $\rho\in[0.3,2.0]$, over 50 independent instances, with \emph{presolve} enabled (left) and disabled (right).
    Solid markers show the mean time to the $\epsilon$-optimal threshold over successful instances, at $\epsilon=0$ and $\epsilon=0.01$, and dashed horizontal bands show the corresponding classical baseline mean $\pm$ standard error.
    Dashed lines with lighter markers show the penalized average time~\eqref{eq:PAR}, which, for instances that never reach the threshold, adds the baseline runtime to the instance's own runtime. Annotations give the number of such instances at each $\rho$.
    All instances reach the $\epsilon=0.01$ threshold.}
    \label{fig:PAR_eps_0}
\end{figure*}

\subsection{Scaling with problem size}
\label{sec:scaling}
We now study how the time to the $\epsilon$-optimal threshold at $\epsilon=0.01$, corresponding to a solution within $1\%$ of the baseline optimum, scales with problem size.
Figure~\ref{fig:scaling_quantum} shows this threshold time as a function of $n$ for the original and preconditioned problems at $p=1,2,3$ and in the $p\rightarrow\infty$ limit, with \emph{presolve} enabled and disabled.
For each instance, we use an oracle strategy to select the penalty parameter $\rho$, among those tested, that minimizes the time to the $\epsilon$-optimal threshold for that instance, and we apply this rule at every problem size and under both \emph{presolve} settings. Appendix~\ref{sec:penalty_strategy} compares this against other selection strategies for $\rho$ such as a globally fixed $\rho$ and an out-of-sample cross-validated $\rho$.
It also reports how the fitted scaling changes with each strategy.

As shown in Figure~\ref{fig:scaling_quantum}, quantum preconditioning reduces the fitted exponential base relative to the classical baseline across the tested range, under both \emph{presolve} settings.
We observe that when using the oracle selection strategy for $\rho$, the fitted bases for $p=1$ through $p=3$ are close to one another and that most of this improvement over the baseline is captured at $p=1$.
At $n=40$, the preconditioned problem at $p=1$ reaches the near-optimal threshold roughly two orders of magnitude sooner than the baseline.
Note that the $p\rightarrow\infty$ limit is the fastest, since its correlation matrix is trivial and built directly from the optimal solution, and we include it as a reference to measure how much of that gap the finite depths close.

Disabling \emph{presolve} steepens the fitted baseline scaling far more than the preconditioned scaling, and the gap between the baseline and preconditioned curves widens further.
These results portray a similar asymmetry as observed in Section~\ref{sec:epstime}

\begin{figure*}[t]
    \centering
    \includegraphics[width=\linewidth]{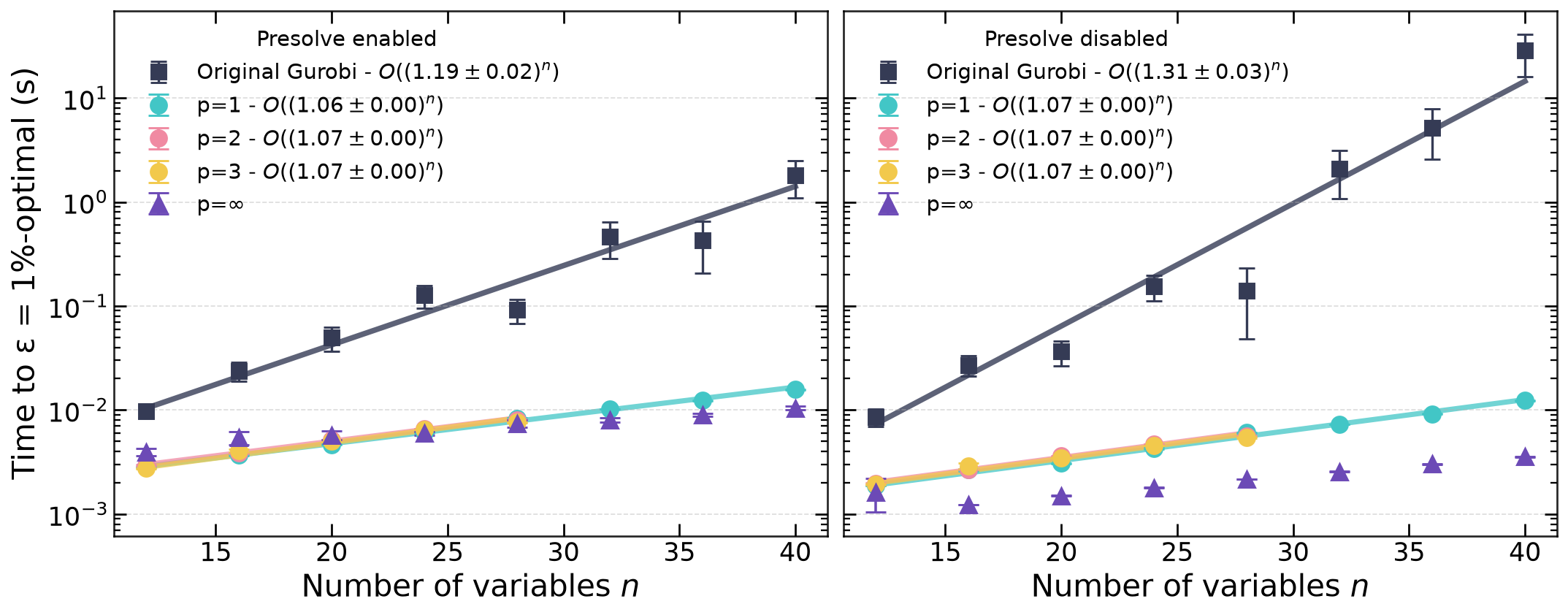}
    \caption{Scaling of the time to the $\epsilon$-optimal threshold at $\epsilon=0.01$ as a function of $n$, with \emph{presolve} enabled (left) and disabled (right).
    The original problem is compared with preconditioned problems at $p=1,2,3$ and at the $p\rightarrow\infty$ limit.
    Each finite-depth curve uses a per-instance oracle penalty, the value of $\rho$ (among those tested) that minimizes the threshold time for that instance.
    Solid lines are exponential fits $\bar t(n)=s a^n$.
    Error bars show the standard error of the mean over 50 instances.}
    \label{fig:scaling_quantum}
\end{figure*}

\subsection{Solution quality distribution across instances}
\label{sec:quality-distribution}
To complement the time-resolved metrics above, we examine the distribution of solution quality achieved by the preconditioned formulations across the full set of tested instances.
As discussed in Section~\ref{sec:callback}, the preconditioned and original objectives coincide exactly only in the $p \to \infty$ limit, so the preconditioned problem's exact minimizer under $Z$ need not equal the baseline optimum under $W$ at finite $p$.
We quantify how often this gap appears in practice and how it behaves as the problem size and circuit depth change.
Figure~\ref{fig:performance_distribution} reports the performance ratio
\begin{equation}
    \pi = \frac{f^\star}{f^\prime} \times 100\%,
\end{equation}
where $f^\prime$ is the best solution found by the preconditioned solve, evaluated under the original objective, and $f^\star$ is the baseline optimum.
Since this is a minimization problem, $f^\prime\ge f^\star$ and $\pi\le100\%$, with $\pi=100\%$ marking exact recovery and lower values marking a suboptimal solution relative to the baseline.

As shown in Figure~\ref{fig:performance_distribution}, exact recovery is common but not universal, and it weakens with depth.
At $p=1$, every tested $n$ recovers the baseline optimum in at least 96\% of instances. At $p=2$ and $p=3$, several sizes still recover it in every instance, but the worst case drops to 94\% at $p=2$ ($n=20$ and $n=28$) and to 88\% at $p=3$ ($n=28$).
The instances that do not recover the baseline optimum still land above 99\% solution quality in every case tested.
We also note that counterintuitively, the widest distribution of solution quality is observed at $p=3$, which also reports the lowest solution quality in the dataset.

\begin{figure*}[t]
    \centering
    \includegraphics[width=0.8\linewidth]{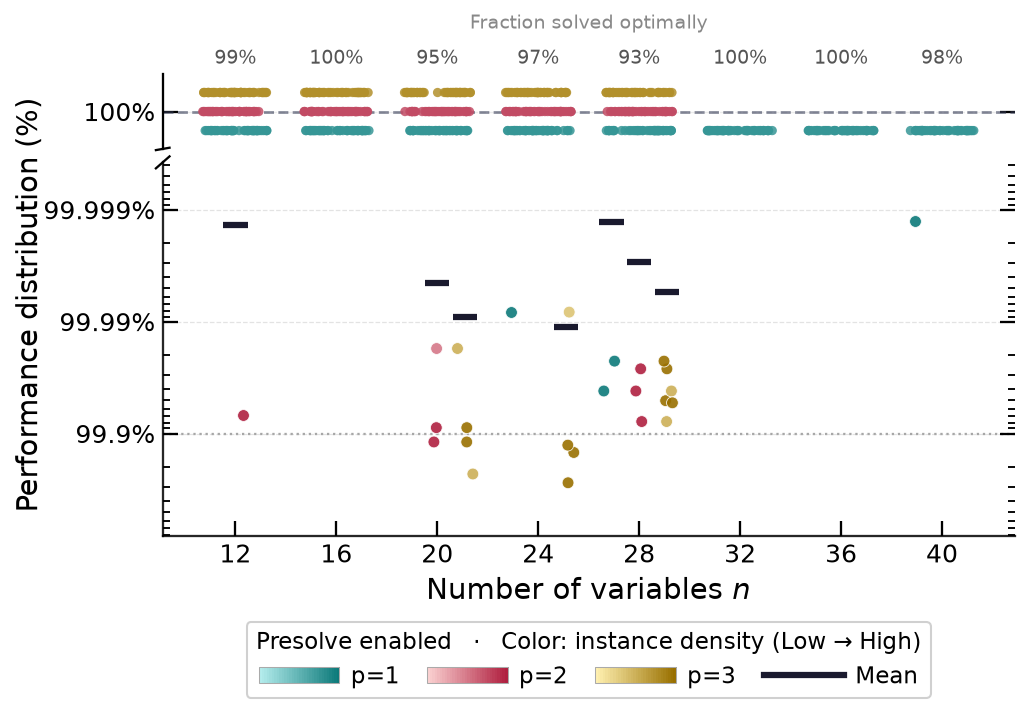}
    \caption{Distribution of solution quality achieved by preconditioned formulations relative to the baseline \texttt{Gurobi} solve, as a function of the number of variables $n$.
    Results are shown for QAOA depths $p = 1, 2, 3$ with 50 problem instances per $n$.
    Points are colored by local instance density within each strip.
    Horizontal bars indicate the mean performance across instances.
    The fraction of instances achieving $\pi = 100\%$ exactly is annotated at the top of the figure.}
    \label{fig:performance_distribution}
\end{figure*}

\section{Discussion}
\label{sec:disc}
We set out to test whether QAOA-derived correlations can serve as external structural information for an exact MIP solver on constrained optimization problems, extending the quantum preconditioning framework of Dupont et al.~\cite{9prw-684p} for unconstrained optimization problems.
Across our results (Section~\ref{sec:results}), we observe that preconditioning helps \texttt{Gurobi} reach an incumbent solution faster than solving the original problem directly, provided a suitable penalty parameter $\rho$ is chosen.

Note that the $\rho$ never enters the \texttt{Gurobi} solve, where the balance constraint is always imposed as a hard constraint (Equation~\eqref{eq:balance_constraint}). 
It acts only inside the QAOA cost Hamiltonian, as a soft-penalty multiplier that weights the balance term against the cut term (Equation~\eqref{eq:cut_objective}).
Therefore, $\rho$ shapes the solve indirectly, through the correlation matrix during preconditioning.
This lets the preconditioned solve recover the original optimum $z^\star$, or a solution within a tolerance of $\epsilon$ relative to the original problem's optimum, as an incumbent faster than the direct solve.
Note that there is no guarantee this happens, which is why we treat $\epsilon$-optimal recovery as an event to be measured. 
The callback trajectory in Figure~\ref{fig:trajectory} shows how this recovery unfolds over the solve across different values of $\rho$.
Near-optimal incumbent discovery is trivial and almost insensitive to $\rho$, since many feasible solutions fall within a percent of the baseline optimum.
Exact recovery, instead, requires the unique optimizer $z^\star$, which the preconditioned solve sometimes fails to recover at small values of $\rho$.
However, as $\rho$ is increased, the fraction of instances where the preconditioned solve recovers the unique optimizer increases.
As a result, the time to recover $z^\star$ as an incumbent solution stays sensitive to $\rho$, which is the spread seen in Figure~\ref{fig:trajectory}. 

Next, we proceed to understand the effect of the \emph{presolve} studies in Figures~\ref{fig:PAR_eps_0} and~\ref{fig:scaling_quantum}.
We observe that when \emph{presolve} is disabled, the baseline solve has a sharp increase in the time to the $\epsilon$-optimal threshold, while the preconditioned solve is barely affected by it. 
In the problem instances we inspected, the root relaxation bound is identical whether \emph{presolve} is enabled or disabled.
Therefore, the increase in time to the $\epsilon$-optimal threshold comes from how \emph{presolve} contributes to the solution search.
We find that this contribution is mainly from the addition of cutting planes, such as implied-bound and cover cuts, which reduce the size of the branch-and-bound tree~\cite{Andersen1995, Gamrath2015presolving}.
The preconditioned solve is largely unaffected by their absence, since the correlation matrix already provides comparable guidance to the solver in finding good incumbent solutions.
 
The scaling results clearly demonstrate the potential of preconditioning.
Across all tested depths, preconditioning lowers the fitted exponential base relative to the baseline.
Moreover, under the oracle selection strategy for $\rho$ shown in Figure~\ref{fig:scaling_quantum}, the fitted preconditioned curves for the three depths cluster closely over the range of problem sizes tested.
We also observe that most of the scaling improvement is already captured at $p=1$ over this range.
Whether this clustering persists at larger problem sizes is not evident from the sizes we consider in this study.
Since deeper circuits carry richer correlations, we expect the curves at different depths to separate at large utility-scale problem sizes that matter in practice.

We observe the reduction in the base of the exponential fits for the preconditioned problem because the QAOA correlations encoded within the correlation matrix $Z$ carry real structure from the original problem even at $p=1$. 
This allows \texttt{Gurobi} to make better branching decisions and discover high-quality incumbents earlier as compared to directly solving the original problem.
At the sizes reachable by exact emulation, this head start is already visible, with $p=1$ reaching the $\epsilon$-optimal threshold of 0.01 roughly two orders of magnitude sooner than the original problem at $n=40$.
This benefit is expected to widen for larger, utility-scale instances where branch-and-bound faces severe challenges, and where finding an efficient path to high-quality incumbent solutions is critical.

We propose several directions worth exploring.
First, the correlations utilized in this work come from classical emulation of the QAOA circuit.
Generating the correlations with the help of quantum hardware, as done for unconstrained problems in~\cite{9prw-684p}, is therefore a natural next step and would highlight the effect of hardware noise and finite sampling on the quality of preconditioning.
Second, comparing quantum preconditioning with strong classical preconditioners, such as deriving correlations between decision variables using simulated annealing, would provide a meaningful performance comparison with an analogous classical approach.
Third, moving from dense to sparse graph families and to more applied constrained problems with multiple constraints, such as the ones provided in QOBLIB~\cite{Koch2026QOBLIB}, would test the versatility of using quantum preconditioning with exact MIP solvers.
We note that this work presents a first step toward quantum preconditioning for constrained optimization, and problems with several constraints are an interesting direction to push next.
Each additional constraint enters the QAOA ansatz as its own penalty term, so the single penalty parameter studied here becomes a set of penalties whose interactions must be tuned together.
Since a poorly chosen penalty parameter already weakens the correlations in a single-constraint case, it could be important to develop systematic ways to set the penalty parameter(s) in order to extend the method to more problems with interesting applications.
An alternative approach is to enforce the constraints within the QAOA ansatz itself, using a constraint-preserving mixing Hamiltonian~\cite{Hadfield2019QAOA, Fuchs2022} rather than a soft penalty. 
This would restrict the correlation matrix to feasible configurations only, and it would be interesting to benchmark how much faster the MIP solver runs on a preconditioned problem constructed this way.
Such mixers guarantee that the QAOA evolution remains within the feasible subspace in the adiabatic limit, but they can be more difficult to implement than standard mixers because the required constraint-preserving dynamics may involve more complex multi-qubit operations and deeper circuit decompositions.
On quantum hardware, this additional circuit complexity can increase exposure to noise and potentially offset the benefit offered by this method of enforcing constraints.
The resulting tradeoff between constraint preservation, circuit complexity, and noise robustness is therefore an important question to investigate.

\paragraph{Practical guidance}
In practice, preconditioning is worth applying when the solver time budget is the binding constraint, which is the regime the scaling results point toward.
For the dense graph instances studied here, $\rho$ can be tuned to satisfy constraints and improve solution quality.
For example, a value at or above the per-instance threshold near $\rho^\star\approx0.35$ gives reliable exact recovery, while smaller values occasionally solve individual instances faster at a higher risk of failure.
Beyond tuning, another practical consideration is the cost paid before the MIP solve, primarily the computation of the correlation matrix.
However, it is anticipated that as the problem size grows to utility-scale, there is a definite window (refer to Equation~\ref{eq:preconditioning_time_model}) where the cost of actually solving the problem directly exceeds the combined cost of computing the correlation matrix and solving the preconditioned problem, making preconditioning a useful strategy at the scales that matter most in practice.

\section{Conclusions}
\label{sec:conclusions}
In this work, we introduced a quantum preconditioning framework for constrained optimization problems.
The preconditioner uses QAOA to generate two-point correlations which are subsequently used to build a modified optimization objective for an exact MIP solver while retaining the hard constraints.
For the problem instances studied here, a suitable choice of the penalty parameter $\rho$ lets the preconditioned problem reach a solution within $1\%$ of the original optimum faster than solving the original problem directly.
We attribute this improvement to the earlier discovery of high-quality incumbents during \texttt{Gurobi}'s branch-and-bound search.
These findings support a hybrid optimization framework in which quantum algorithms act as structure-learning pre-processors that supply problem-specific guidance to classical exact MIP solvers.

\appendix

\section{Penalty threshold analysis and proofs}
\label{sec:penalty_proofs}

The graph that witnesses the sharpness of Theorem~\ref{thm:det_threshold} has weights exactly in $\{0,2\}$, a probability-zero event under continuous sampling.
The next result shows that randomness does not lower the smallest penalty that guarantees feasibility with probability one, though it does not by itself quantify how frequently failure occurs below that guarantee.

\begin{theorem}[Threshold for an almost-sure guarantee]
\label{thm:as_threshold}
Let the edge weights $w_{ij}$ be drawn i.i.d.\ from $U[0, 2]$.
For every $\rho < n/4$, the event that every minimizer of~\eqref{eq:penalized_spin} is infeasible has positive probability.
At $\rho = n/4$, every minimizer of~\eqref{eq:penalized_spin} is feasible with probability one.
Hence $\rho = n/4$ is the sharp threshold for a uniform almost-sure exactness guarantee under i.i.d.\ $U[0,2]$ weights.
\end{theorem}

Theorem~\ref{thm:as_threshold} rules out the natural conjecture that switching from adversarial to i.i.d.\ weights lowers the required penalty.
The worst-case value $n/4$ remains the sharp threshold almost surely, so the gap to the empirical typical value discussed next is genuinely an average-versus-tail phenomenon rather than an artifact of the worst-case bound.

\paragraph{Empirical per-instance threshold}
\label{prop:typical_threshold}
Define the per-instance critical penalty $\rho^\star := \inf\{\rho > 0 : \text{every minimizer of~\eqref{eq:penalized_spin} is feasible}\}$.
For complete graphs with i.i.d.\ $U[0,2]$ weights, we evaluated $\rho^\star$ by exhaustive enumeration at $n\in\{8,12,16,20\}$ using 50 instances at each size.
Over this finite sample, the medians lie between $0.34$ and $0.36$ and the 99th percentiles between $0.63$ and $0.89$ in spin units.
These observations are numerical and are not asserted as an average-case theorem.
The distribution is plotted in Figure~\ref{fig:penalties_per_instance}.

The ``Max'' curve in Figure~\ref{fig:penalties_per_instance} is the largest $\rho^\star$ among the 50 instances at each $n$, not the essential supremum over all admissible weights, and the calculation uses the same seed-indexed instance family as the MIP scaling experiments.

The behavior of a typical instance can be further illuminated by computing the expected effect of a single repair step, which formalizes the intuition that for large enough $\rho$ the penalty term dominates and pushes any imbalanced configuration toward feasibility.

\begin{proposition}[Expected repair cost]
\label{prop:expected_repair}
Let $|S(\mathbf{z})| = m + r$ with $r \geq 1$, so that $\mathbf{z}$ is infeasible with imbalance $2r$.
Let the edge weights be i.i.d.\ $U[0,2]$.
A single repair step that moves one vertex out of $S$ produces an expected change in cut value $\mathbb{E}[\Delta H] = 2r - 1$ and an expected change in the penalized objective
\begin{equation}
    \mathbb{E}[\Delta C] = (2r-1)(1 - 4\rho),
    \label{eq:expected_repair}
\end{equation}
which is strictly negative for every $\rho > 1/4$.
\end{proposition}

\begin{remark}
\label{rem:average_case_caveat}
Proposition~\ref{prop:expected_repair} concerns the expected change for a single repair step from a \emph{fixed} configuration, not the behavior of the penalized minimizer.
Two cautions follow.
First, a negative expected change does not imply that every realization improves, since for $r = 1$ the standard deviation of $\Delta H$ grows as $\Theta(\sqrt{n})$, so the repair step does not concentrate at any fixed $\rho$.
Second, exactness concerns the configuration selected by the penalized optimizer, whose choice is correlated with the weights, so the independence used in~\eqref{eq:expected_repair} does not apply in that setting.
Proposition~\ref{prop:expected_repair} therefore motivates, but does not establish, a small typical threshold.
Figure~\ref{fig:penalties_per_instance} provides the complementary empirical characterization.
\end{remark}

These results together characterize three distinct regimes of the penalty parameter.
The worst-case and almost-sure thresholds of Theorems~\ref{thm:det_threshold} and~\ref{thm:as_threshold} scale as $n/4$, reaching values of $3$ and $10$ at $n = 12$ and $n = 40$ respectively.
The expected single-step repair of Proposition~\ref{prop:expected_repair} is favorable for any $\rho > 0.25$, independent of $n$.
The empirical per-instance study has medians near $0.35$ and 99th percentiles below $0.90$ over the tested sizes.
These summaries answer different questions about the same quantity.
The worst-case and almost-sure thresholds describe the essential supremum of $\rho^\star$.
The expected-repair threshold describes a favorable-in-expectation boundary.
The empirical distribution describes the bulk.
Together they establish $\rho \in [0,2]$ as a well-motivated study range because it covers nearly all of the sampled per-instance distribution while remaining far below the worst-case threshold.
The sensitivity to $\rho$ observed in Section~\ref{sec:results} can therefore be understood as a consequence of the gap between these regimes and the solver's need for a sufficiently strong correlation signal.

We prove Theorems~\ref{thm:det_threshold} and~\ref{thm:as_threshold} and Proposition~\ref{prop:expected_repair}, then describe the exhaustive per-instance computation.
Throughout, $n = |V|$ is even with $n = 2m$, and we work with spin variables $\mathbf{z} \in \{-1,+1\}^n$ and the penalized objective $C(\mathbf{z})$ defined in~\eqref{eq:penalized_spin}.
We denote $S(\mathbf{z}) := \{i : z_i = +1\}$, write $k = |S|$, and let $R_i^-(\mathbf{z})$ denote the configuration obtained by flipping $z_i$ from $+1$ to $-1$ for $i \in S(\mathbf{z})$, and $R_i^+(\mathbf{z})$ the flip of $z_i$ from $-1$ to $+1$ for $i \notin S(\mathbf{z})$.

\begin{remark}[Relationship to binary variables]
\label{rem:binary}
Setting $x_i = (z_i+1)/2 \in \{0,1\}$, one has $H(\mathbf{z}) = C(\mathbf{x}) := \sum_{\{i,j\} \in E}w_{ij}(x_i + x_j - 2x_ix_j)$ and $(\sum_i z_i)^2 = 4(\sum_i x_i - m)^2$.
Hence $C(\mathbf{z}) = C(\mathbf{x}) + \rho_x(\sum_i x_i - m)^2$ with $\rho_x = 4\rho$.
Every result below stated in spin units $\rho$ has an equivalent binary-variable statement with threshold $\rho_x = 4\rho$.
\end{remark}

\begin{lemma}[Complement symmetry]
\label{lem:symmetry}
For $\bar{\mathbf{z}} := -\mathbf{z}$, one has $C(\bar{\mathbf{z}}) = C(\mathbf{z})$ for every $\rho \geq 0$.
Minimizers therefore occur in complementary pairs and no minimizer is unique.
\end{lemma}

\begin{proof}
Since $\bar{z}_i\bar{z}_j = z_iz_j$, the cut value $H$ is unchanged.
Since $\sum_i \bar{z}_i = -\sum_i z_i$, the penalty is unchanged.
\end{proof}

\begin{proof}[Proof of Theorem~\ref{thm:det_threshold}]
On feasible points the penalty vanishes, so $\min_{\mathbf{z}} C \leq H^\star$, where $H^\star$ is the constrained optimum.
Suppose $\tilde{\mathbf{z}}$ minimizes $C$ and is infeasible, with $k = |S(\tilde{\mathbf{z}})| = m+r$ for some $r \geq 1$.
Pick $i \in S(\tilde{\mathbf{z}})$ and let $\mathbf{z}' = R_i^-(\tilde{\mathbf{z}})$.
Flipping $z_i$ from $+1$ to $-1$ converts its edges to $S \setminus \{i\}$ from non-cut to cut and its edges to $V \setminus S$ from cut to non-cut, so
\begin{equation}
    H(\mathbf{z}') - H(\tilde{\mathbf{z}})
    = \sum_{j \in S \setminus \{i\}} w_{ij}
    - \sum_{j \notin S} w_{ij}
    \leq 2(k-1) = n + 2r - 2.
    \label{eq:cut_change_app}
\end{equation}
The penalty changes from $4\rho r^2$ to $4\rho(r-1)^2$, a reduction of $4\rho(2r-1)$, giving
\begin{equation}
    C(\mathbf{z}') - C(\tilde{\mathbf{z}})
    \leq n + 2r - 2 - 4\rho(2r-1).
    \label{eq:total_change_app}
\end{equation}
For $\rho > n/4$ and $r \geq 1$, the right side of~\eqref{eq:total_change_app} is strictly negative, contradicting the optimality of $\tilde{\mathbf{z}}$.
If $k=m-r$, choose $i\notin S$ and apply $R_i^+$.
The analogous cut change has the same upper bound because there are $m+r-1$ vertices on the current majority side and $m-r$ on the minority side.
Hence every minimizer of $C$ is feasible, and on the feasible set $C = H$, so the minimizers of~\eqref{eq:penalized_spin} are exactly those of the original constrained problem.

\emph{Sharpness.}
Partition $V = A \cup B$ with $|A| = m-1$, $|B| = m+1$, set $w_{ij} = 2$ for all distinct $i,j \in B$, and $w_{ij} = 0$ otherwise.
The infeasible configuration $\mathbf{z}^A$ with $S = A$ satisfies $H(\mathbf{z}^A) = 0$ and penalty $4\rho \cdot 1 = 4\rho$, giving $C(\mathbf{z}^A) = 4\rho$.
Any feasible configuration has at least one vertex from $B$ in $S$, contributing cut weight at least $2m = n$, so $H^\star = n$.
For $\rho < n/4$, $C(\mathbf{z}^A) = 4\rho < n = H^\star$, so the infeasible $\mathbf{z}^A$ beats every feasible point.
\end{proof}

\begin{proof}[Proof of Theorem~\ref{thm:as_threshold}]
\emph{Positive-probability failure below $n/4$.}
Fix $\rho < n/4$ and choose $\delta > 0$ small enough that $4\rho + \delta(m^2-1) < n - \delta m$.
On the event $\mathcal{E}_\delta$ that all $A$--$B$ weights lie in $[0, \delta]$ and all $B$-internal weights lie in $[2-\delta, 2]$, the infeasible configuration $\mathbf{z}^A$ has $C(\mathbf{z}^A) < H^\star$ by the same estimate as the sharpness argument above.
Since the weights are continuous, $\mathbb{P}(\mathcal{E}_\delta) > 0$.

\emph{Almost-sure exactness at $\rho = n/4$.}
Suppose $\tilde{\mathbf{z}}$ minimizes $H_{n/4}$ and is infeasible.
By~\eqref{eq:total_change_app} at $\rho = n/4$, any imbalance of $r \geq 2$ units can be repaired with a non-positive change in $H_{n/4}$, strict for $r \geq 2$.
Hence a minimizer satisfies $|\sum_i z_i| = 2$ (one unit off balance, i.e., $|k - m| = 1$).
Optimality at $r = 1$ forces equality in~\eqref{eq:cut_change_app}, which requires all weights $w_{ij}$ for $j \in S \setminus \{i\}$ to equal $2$ and all weights $w_{ij}$ for $j \notin S$ to equal $0$.
Under continuous i.i.d.\ weights this event has probability zero, and the union over the finitely many pairs $(S, i)$ is also a null event.
\end{proof}

\subsection{Per-instance critical-penalty computation}
Because the configuration space is finite, the per-instance critical penalty satisfies the exact identity
\begin{equation}
    \rho^\star
    = \max\left\{0,\,
    \max_{\mathbf{z}\ {\rm infeasible}}
    \frac{H^\star - H(\mathbf{z})}{(\sum_i z_i)^2}
    \right\}.
    \label{eq:rho_star_sup}
\end{equation}
For each $n\in\{8,12,16,20\}$, we use the 50 instances with integer seeds $0,\ldots,49$.
For seed $s$, we initialize \texttt{numpy.random.default\_rng($s$)}, draw $n(n-1)/2$ independent values from $U[0,2]$, and assign them to the strict upper triangle in \texttt{NumPy} index order.
We then evaluate all configurations, compute $H^\star$ over the balanced configurations, and apply \eqref{eq:rho_star_sup}.
The resulting distribution is summarized in Figure~\ref{fig:penalties_per_instance}.
The observed near-constant quantiles are empirical and are not a proof of size-independent concentration.

\begin{proof}[Proof of Proposition~\ref{prop:expected_repair}]
With $k = m + r$, the change in cut value upon removing vertex $i \in S$ is
\begin{equation}
    \Delta H = \sum_{j \in S \setminus \{i\}} w_{ij}
             - \sum_{j \notin S} w_{ij}.
\end{equation}
This is a sum of $k - 1 = m + r - 1$ terms of mean $1$ minus $n - k = m - r$ terms of mean $1$, giving $\mathbb{E}[\Delta H] = (m+r-1) - (m-r) = 2r - 1$.
The penalty changes from $\rho(\sum_i z_i)^2 = 4\rho r^2$ to $4\rho(r-1)^2$ deterministically, a reduction of $4\rho(2r-1)$.
Hence
\begin{equation}
    \mathbb{E}[\Delta C]
    = \mathbb{E}[\Delta H] - 4\rho(2r-1)
    = (2r-1) - 4\rho(2r-1)
    = (2r-1)(1 - 4\rho),
\end{equation}
which is strictly negative for $\rho > 1/4$ and any $r \geq 1$.
If $k=m-r$, choose $i\notin S$ and flip it from $-1$ to $+1$.
There are $m+r-1$ other vertices outside $S$ and $m-r$ vertices in $S$, so
\begin{equation}
\mathbb E[\Delta H]=(m+r-1)-(m-r)=2r-1.
\end{equation}
The penalty reduction is again $4\rho(2r-1)$.
Thus \eqref{eq:expected_repair} and the threshold $\rho>1/4$ hold on both sides of balance.
\end{proof}

\section{Penalty Parameter selection strategy and scaling}
\label{sec:penalty_strategy}

The scaling analysis of Section~\ref{sec:scaling} selects, for each QAOA depth $p$ and each instance, the value of $\rho$ that minimizes that instance's own time to the $\epsilon$-optimal threshold, an oracle rule that requires knowing the outcome in advance and serves as an upper bound on what preconditioning can achieve.
This appendix asks how much of that reported benefit survives once $\rho$ has to be chosen without this hindsight, using two distinct selection rules described below.

\paragraph{Global $\rho$ selection}
For each QAOA depth $p$, the global rule selects a single value of $\rho$ that minimizes the mean time to $\epsilon$-optimal threshold for each instance $s \in \mathcal S$ per $n$  averaged across all tested $n$.
This rule requires no instance dependent tuning and could be deployed in practice, choosing $\rho$ once in advance.

\paragraph{Cross-validated $\rho$ selection}
For each instance $s \in \mathcal S$ at problem size $n$, the cross-validated rule selects the value of $\rho$ that minimizes the mean time to $\epsilon$-optimal threshold over the $\mathcal S \setminus s$ instances at the same $n$

Figure~\ref{fig:scaling_panel_presolve_on} reports the resulting time to $\epsilon$-optimal threshold scaling under both rules at the near-optimal tolerance $\epsilon=0.01$ used in Section~\ref{sec:scaling}, with \emph{presolve} enabled (top row) and disabled (bottom row).

Under the two \emph{presolve} settings, we observe that both the strategies give larger fitted bases than the oracle strategy reported in Section~\ref{sec:scaling} at every depth.
However, unlike the oracle result, where the three depths cluster closely together, we observe a clear separation among the curves corresponding to the different depths.
$p=1$ sits much closer to the classical baseline while $p=2$ and $p=3$ pull well ahead of it.
Cross-validated selection improves on global selection at $p=1$ and $p=3$, but global selection is slightly better at $p=2$, so neither rule dominates the other at each depth considered in this study.

Despite this, both the global and cross-validated curves remain below the classical baseline at every depth and in both \emph{presolve} configurations.
We also observe that the qualitative comparison between the two selection strategies is similar whether \emph{presolve} is enabled or disabled, indicating that the \emph{presolve}-dependent asymmetry identified in Section~\ref{sec:epstime} is likewise not an artifact of how $\rho$ is chosen.

\begin{figure*}[!t]
    \centering
    \includegraphics[width=\linewidth]{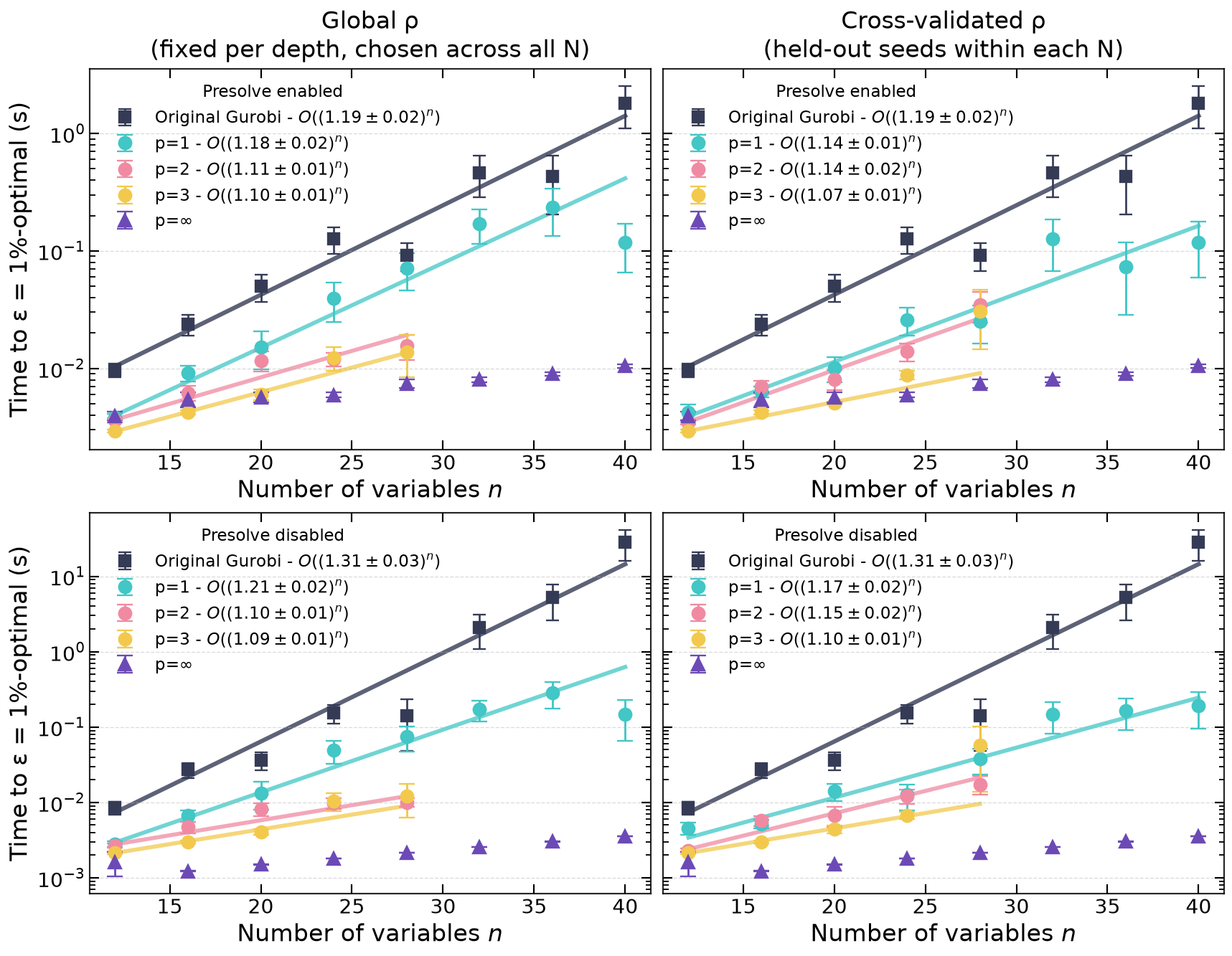}
    \caption{Sensitivity of the time to $\epsilon$-optimal threshold scaling at $\epsilon=0.01$ to penalty selection, with \emph{presolve} enabled (top row) and disabled (bottom row). 
    The left column uses a single global-$\rho$ per depth, chosen across all problem sizes.
    The right column selects $\rho$ by cross-validation on held-out instances within each $n$.
    Points show means with standard errors, and solid lines show exponential fits for all depths}
    \label{fig:scaling_panel_presolve_on}
\end{figure*}

\section{Parameter-transfer validation via QAOA landscapes}
\label{sec:landscape_validation}

Section~\ref{sec:quantum_preconditioning} transfers QAOA angles to problem sizes $n>20$ from angles optimized at a reference size $n_{\rm ref}=20$, using $\beta_\ell(n)=\beta_\ell(n_{\rm ref})$ and $\gamma_\ell(n)=\gamma_\ell(n_{\rm ref})\sqrt{n_{\rm ref}/n}$.
This appendix illustrates that rule for a single representative instance of balanced graph bi-partitioning ($\rho=1.2$, seed $0$), by comparing a 2D slice of the QAOA cost landscape at the reference size against the same slice at each transferred size.
This is a per-instance illustration of whether a transferred point lands in the same basin, not a statistical validation averaged over instances.

At $p=1$, the angle space is two-dimensional, so Figure~\ref{fig:landscape_combined}(a) plots the full landscape $\braket{\hat C(\gamma_1,\beta_1)}$ directly.
At $p=2$ and $p=3$, the angle space is four- and six-dimensional, so a full landscape cannot be plotted; panels (b) and (c) instead show a reproducible 2D slice that varies only the first layer's $(\gamma_1,\beta_1)$ over the same rescaled grid, while every other layer's angle is held fixed at its actual optimized (at $n_{\rm ref}$) or transferred (at $n>n_{\rm ref}$) value.
In all three panels, $\braket{\hat C(\gamma,\beta)}$ is minimized during optimization, so the dark, low-value regions mark favorable angles, and each panel is normalized independently to $[0,1]$ over its own minimum and maximum.
The leftmost panel in each row shows the reference size $n_{\rm ref}=20$, with a star marking the optimized $(\gamma_1,\beta_1)$ found by the procedure detailed in Section~\ref{sec:quantum_preconditioning}.
The remaining panels show the transferred sizes, with the horizontal axis rescaled by $\sqrt{n/n_{\rm ref}}$.
The cross in each panel marks the transferred angle obtained by applying the rescaling rule above to the reference optimum.
At every tested size and at all three depths, the low-cost (dark) region sits at the same rescaled location across panels, and the transferred marker falls inside or immediately adjacent to it, illustrating that the rescaling rule continues to locate a near-optimal region of this slice as $n$ grows for this instance.

\begin{figure*}[!t]
    \centering
    \begin{subfigure}[!t]{\linewidth}
        \centering
        \includegraphics[width=\linewidth]{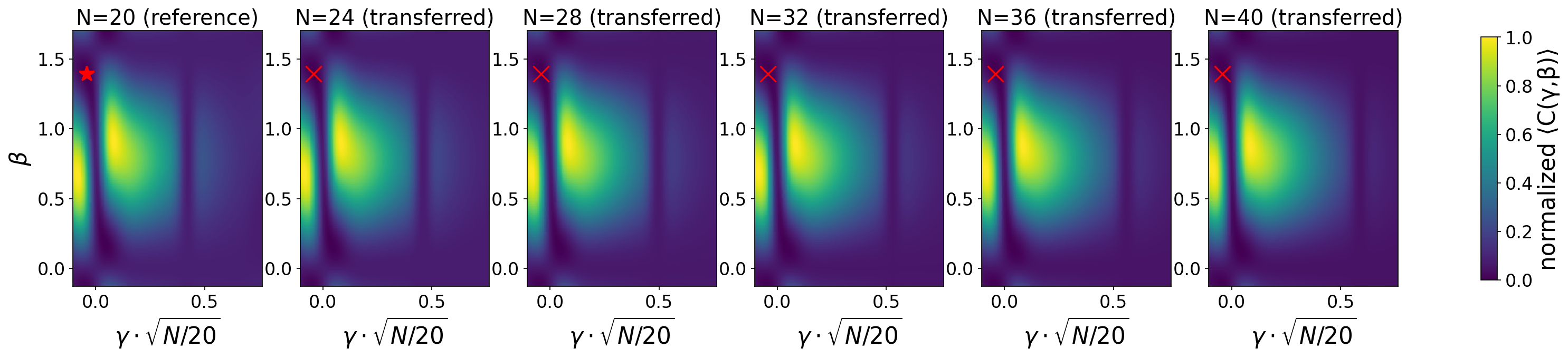}
        \caption{$p=1$: reference size $n_{\rm ref}=20$ and transferred sizes $n=24,28,32,36,40$.}
        \label{fig:landscape_p1}
    \end{subfigure}
    \\[8pt]
    \begin{subfigure}[!t]{0.54\linewidth}
        \centering
        \includegraphics[width=\linewidth]{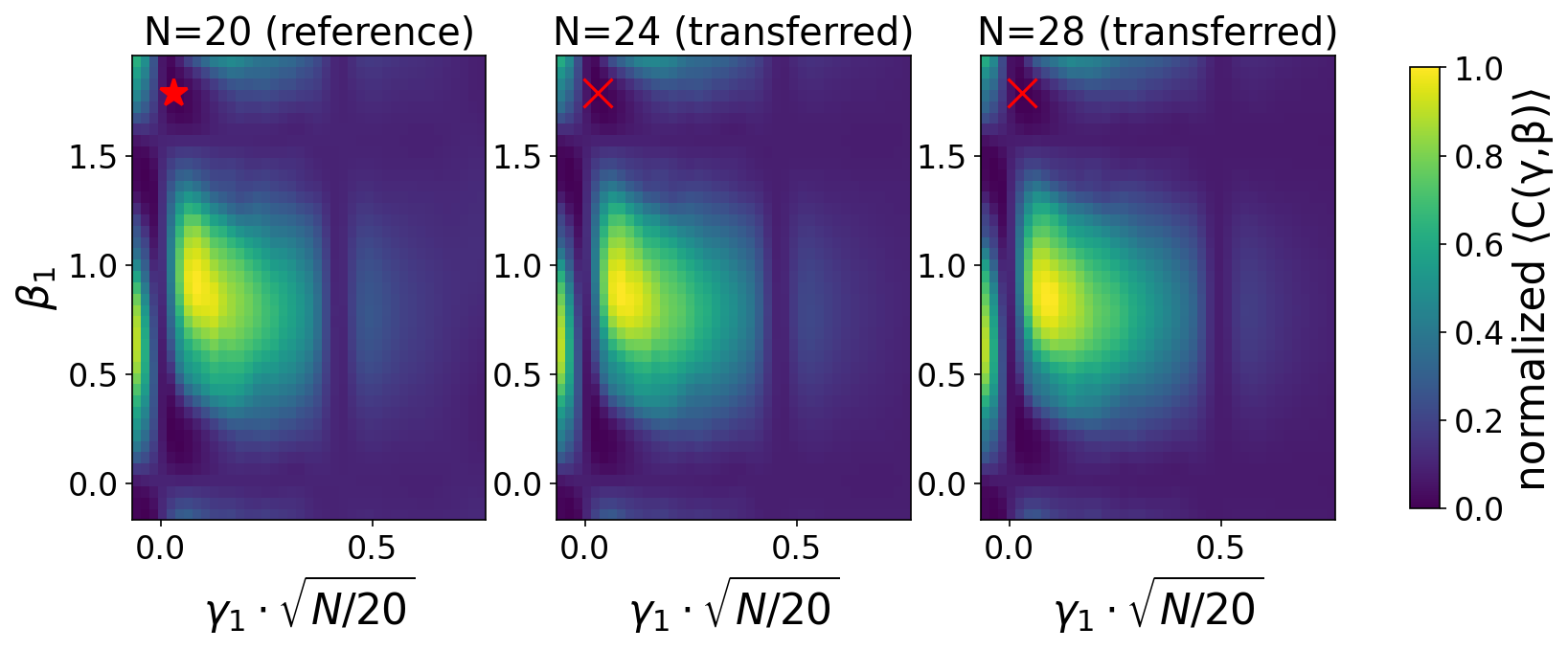}
        \caption{$p=2$: reference size $n_{\rm ref}=20$ and transferred sizes $n=24,28$.}
        \label{fig:landscape_p2}
    \end{subfigure}
    \\[8pt]
    \begin{subfigure}[!t]{0.54\linewidth}
        \centering
        \includegraphics[width=\linewidth]{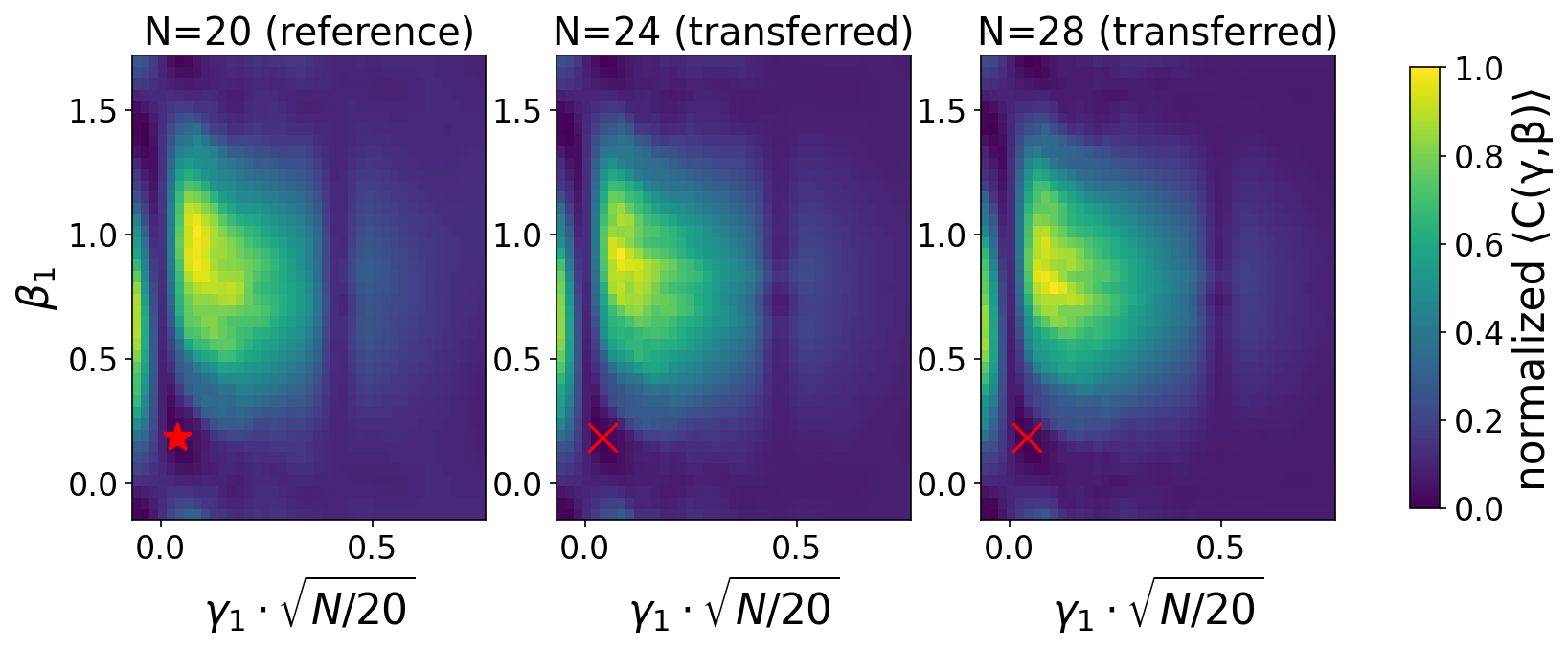}
        \caption{$p=3$: reference size $n_{\rm ref}=20$ and transferred sizes $n=24,28$.}
        \label{fig:landscape_p3}
    \end{subfigure}
    \caption{QAOA landscape $\braket{\hat C(\gamma,\beta)}$ for a single representative instance ($\rho=1.2$, seed $0$) at depths (a)~$p=1$, (b)~$p=2$, and (c)~$p=3$.
    Each panel is normalized independently to $[0,1]$ over its own minimum and maximum.
    Panel (a) shows the full $p=1$ landscape. 
    Panels (b) and (c) show a 2D slice varying only the first layer's $(\gamma_1,\beta_1)$ with all other layers' angles held fixed at their actual optimized or transferred values.
    In each row, the leftmost panel shows the reference size $n_{\rm ref}=20$, with a star marking the optimized angle, and the remaining panels show angles transferred from the reference optimum using the rescaling rule of Section~\ref{sec:quantum_preconditioning}.
    The horizontal axis is rescaled by $\sqrt{n/n_{\rm ref}}$ so that the transferred marker (cross) is directly comparable to the reference star.}
    \label{fig:landscape_combined}
\end{figure*}

\section*{Reproducibility}
The code and data that generate all figures and numerical results in this paper are available at \url{https://github.com/SECQUOIA/Quantum_Preconditioning_Constrained_Optimization}. The repository includes the instances, the correlation matrices, the penalty-threshold computations of Appendix~\ref{sec:penalty_proofs}, and the scripts that produce each figure.

\section*{Acknowledgments}
Authors A.R. and D.E.B.N. acknowledge support by the Industry-University Cooperative Research Center Program at the US National Science Foundation under Grant No. 2224960. Author M.D. acknowledges support by the U.S. Department of Energy, Office of Science, National Quantum Information Science Research Centers, Superconducting Quantum Materials and Systems Center (SQMS), under Contract No. 89243024CSC000002. Fermilab is operated by Fermi Forward Discovery Group, LLC under Contract No. 89243024CSC000002 with the U.S. Department of Energy, Office of Science, Office of High Energy Physics.

\newpage

\bibliography{references}
\end{document}